\documentclass[10pt,letterpaper,conference,twocolumn]{IEEEtran}

\usepackage{graphicx}
\usepackage{amsmath}
\usepackage{amssymb}
\usepackage{amsxtra}
\usepackage{epsfig}
\usepackage{amsxtra}
\usepackage{amsfonts}
\usepackage{url}

\usepackage[active]{srcltx}
\usepackage{epsf,graphicx}
\usepackage{bm}

\newtheorem{thm}{Theorem}
\newtheorem{prop}{Proposition}

\newtheorem{lem}{Lemma}

\usepackage{amsmath}
\usepackage{latexsym}
\usepackage{enumerate}
\usepackage{graphics}
\usepackage{graphicx}
\usepackage{algorithm}
\usepackage{algorithmic}
\usepackage{psfrag}
\usepackage{amsfonts}
\usepackage{amssymb}%
\providecommand{\U}[1]{\protect\rule{.1in}{.1in}}
\providecommand{\U}[1]{\protect\rule{.1in}{.1in}}
\providecommand{\U}[1]{\protect\rule{.1in}{.1in}}
\providecommand{\U}[1]{\protect\rule{.1in}{.1in}}

\begin{document}

\title{On Interference Alignment and the Deterministic Capacity for Cellular Channels with Weak Symmetric Cross Links}

\author{
\authorblockN{J\"org B\"uhler}
\authorblockA{Heinrich-Hertz-Lehrstuhl f\"ur Informationstheorie\\ und theoretische Informationstechnik \\
Technische Universit\"at Berlin\\
Einsteinufer 25, D-10587 Berlin, Germany\\
Email: joerg.buehler@mk.tu-berlin.de}
\and
\authorblockN{Gerhard Wunder}
\authorblockA{Fraunhofer Heinrich Hertz Institute \\
Einsteinufer 37, D-10587 Berlin, Germany\\
Email: gerhard.wunder@hhi.fraunhofer.de}
}

\maketitle
\def\thefootnote{\fnsymbol{footnote}}
\footnotetext{This research is supported by Deutsche Forschungsgemeinschaft (DFG) under grant WU 598/1-2.}

\renewcommand{\thefootnote}{\arabic{footnote}} 

\begin{abstract}
In this paper, we study the uplink of a cellular system using the linear deterministic approximation model, where there are two users transmitting to a receiver, mutually interfering with a third transmitter communicating with a second receiver. 
We give an achievable coding scheme and prove its optimality, i.e. characterize the capacity region. This scheme is a form of interference alignment which exploits the channel gain difference of the two-user cell. 
\end{abstract}

\section{Introduction}\label{sec:introduction} 
In recent years, approximate characterizations of capacity regions of multi-user systems have gained more and more attention, with one of the most prominent examples being the characterization of the capacity of the Gaussian interference channel to within one bit/s/Hz in \cite{EtkinTseWangIT08}. One of the tools that arised in the context of capacity approximations and has been shown to be useful in many cases is the {\em linear deterministic model} introduced in \cite{AvestimehrTseAllerton2007} \cite{AvestimehrDiggaviTse_IT2011}. Here, the channel is modelled as a deterministic mapping that operates on bit vectors and mimics the effect the physical channel and interfering signals have on the binary expansion of the transmitted symbols. Basically, the effect of the channel is to erase a certain number of ingoing bits, while superposition of signals is given by the modulo addition. 
Even though this model deemphasizes the effect of thermal noise, it is able to capture some important basic features of wireless systems, namely the superposition and broadcast properties of electromagnetic wave propagation.  Hence, in multi-user systems where interference is one of the most important limiting factors on system performance, the model can also be useful to devise effective coding and interference mitigation techniques. There are many examples where a linear deterministic analysis can be successfully carried over to coding schemes for the physical (Gaussian) models or be used for approximative capacity or degrees of freedom determination. Among them are results concerning the the two-user interference channel \cite{BreslerTseETCOMM08}, the $K$-user interference channel \cite{Cadambe_IT2009}, the $X$-channel \cite{HuangCadambeJafarISIT09}, fading broadcast channels without channel state information at the transmitter \cite{TseYates09}, many-to-one and one-to-many interference channels \cite{Bresler_IT2010}, cooperative interference channels \cite{PrabhakaranViswanath_IT2011Source,PrabhakaranViswanath_IT2011Destination} and cyclically symmetric deterministic interference channels \cite{Bandemer_ISIT2010}. Here, the optimal transmission schemes often involve {\em interference alignment}, where the transmitted signals are designed such that at the receivers, the undesired part of the signal that is due to the different interfering signals aligns in a certain subspace of the receive space.

{\bf Contributions.}
From a practical viewpoint, {\em cellular systems} are of major interest. Generally, a cellular system consists of a set of base stations each communicating with a distinct set of (mobile) users. Effective coding and interference mitigation schemes are still an active area of research. Approximative models such as the linear deterministic approach might help to gain more insight into these problems. 
In this paper, we take a step into this direction and investigate a cellular setup using the linear deterministic model. To our knowledge, this is the first linear deterministic capacity analysis for cellular-type channels. We determine the capacity region and the corresponding optimal coding scheme. This  scheme involves a form of interference alignment, exploiting the channel gain difference of the two-user cell.

{\bf Organization.}
The paper is organized as follows: 
Section \ref{sec:systemmodel} introduces the system model. In section \ref{sec:achievable region}, we give an achievable
rate region and the corresponding communication scheme, which essentially uses the principle of interference alignment. 
In section \ref{sec:capacity}, we prove that this region actually constitutes the capacity region by deriving an outer bound region that coincides with the achievable rate region. 
Finally, section \ref{sec:conclusions} concludes the paper.

{\bf Notation.}
Throughout the paper, $\mathbb{F}_2 = \{0,1\}$ denotes the binary finite field, for which addition is written as $\oplus$, which is addition modulo 2. For two matrices $A \in \mathbb{F}_2^{n_A \times m}$ and $B \in \mathbb{F}_2^{n_B \times m}$, we denote by $[A;B] \in \mathbb{F}_2^{n_A + n_B \times m}$ the matrix that is obtained by stacking $A$ over $B$. Moreover, for a sequence of matrices $A_1, \ldots, A_n$, we write $(A_k)_{k=1}^n$ for the stacked matrix, i.e.
$(A_k)_{k=1}^n = [(A_k)_{k=1}^{n-1}; A_n]$. Similarily, $[A|B]$ stands for placing $A$ next to $B$. The zero matrix of size $m \times n$ and the matrix with all entries equal to one are denoted by $\mathbf{0}_{m\times n}$ and $\mathbf{1}_{m\times n}$, respectively. Furthermore, $\mathbf{1}_m = \mathbf{1}_{m\times 1}$, $\mathbf{0}_m = \mathbf{0}_{m\times 1}$ and $I_{n}$ is the $n \times n$ identity matrix. For a matrix $A \in \mathbb{F}_2^{m \times n}$, we write $|A|$ for the number of ones in $A$. Finally, $\text{div}$ and $\text{mod}$ denote integer division and the modulo operation, respectively, where we use the convention $x \text{~div~} 0 = 0$.

\section{System Model}\label{sec:systemmodel} 

\begin{figure}
\begin{center}
\includegraphics[scale=0.55]{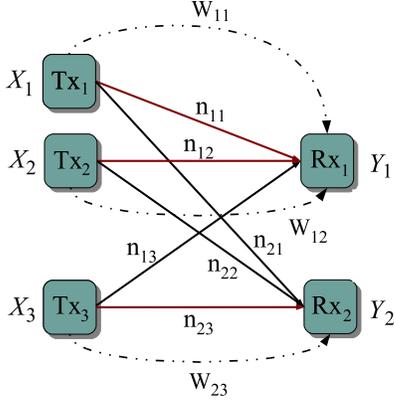}
\caption{System model.}
\label{fig:systemmodel}
\end{center} 
\end{figure} 

The system we consider here represents a basic version of the uplink of cellular system and consists of three transmitters (mobile users) $Tx_{1},Tx_{2}$ and $Tx_{3}$ and two receivers (base stations) $Rx_{1}, Rx_{2}$ (c.f. Figure \ref{fig:systemmodel}). The system is modelled using the {\em linear deterministic model} \cite{AvestimehrTseAllerton2007, AvestimehrDiggaviTse_IT2011}. 
Here, the input symbol at transmitter $Tx_{i}$ is given by a bit vector $X_i \in \mathbb{F}_2^q$ and the output bit vectors $Y_j$ at $Rx_{j}$ are deterministic functions of the inputs: 
Defining the shift matrix $S \in \mathbb{F}_2^{q \times q}$ by
\begin{equation}
S = \begin{pmatrix} 0 & 0 & 0 & \cdots & 0\\ 1 & 0 & 0 & \cdots & 0\\ 0 & 1& 0 & \cdots &0 
\\ \vdots & \ddots& \ddots & \ddots & \vdots \\ 0 & \cdots& 0 &1 & 0 \\ 
\end{pmatrix},
\end{equation}
the input/output equations of the system are given by
\begin{IEEEeqnarray}{rCl}
Y_1 & = & S^{q-n_{11}} X_1 \oplus S^{q-n_{12}}X_2 \oplus S^{q-n_{13}}X_3, \\ \notag
Y_2 & = & S^{q-n_{21}} X_1 \oplus S^{q-n_{22}}X_2 \oplus S^{q-n_{23}}X_3. 
\end{IEEEeqnarray}
Here, $q$ is chosen arbitrarily such that $q \geq \text{max}_{i,j}\{n_{ij}\}$. Note that $n_{ij}$ gives the number of bits that can be passed over the channel between $Tx_{j}$ and $Rx_{i}$, i.e. $n_{ij}$ represent channel gains. There are three messages to be transmitted in the system: $W_{ij}$ denotes the message from transmitter $Tx_{j}$ to the intended receiver $Rx_{i}$. The definitions of (block) codes, error probability, achievable rates and the capacity region are according to the standard information-theoretic definitions. For the remainder of the paper, the transmission rate corresponding to message $W_{11}$ is represented by $R_1$, the rate corresponding to $W_{12}$ by $R_2$ and the rate for $W_{23}$ by $R_3$. 

In the following, we assume without loss of generality that $n_{11} \geq n_{12}$ and write $n_1 = n_{11}, n_2 = n_{12}, n_3 = n_{23}, n_D = n_{13}$ and $\Delta = n_1 - n_2$. In this paper, we further restrict ourselves to the case $n_{21} = n_{22} =: n_{M}$, to which we refer to as {\em symmetric MAC interference} and assume the {\em weak interference condition} $n_D + n_M \leq \text{min}\left\{n_2,n_3\right\}$. Also, we may choose $q = \text{max} \{n_1,n_3\}$. 

\section{An Achievable Region}\label{sec:achievable region} 
In this section, we derive an achievable rate region. For this, we split the system into two subsystems and derive achievable regions for each of them, denoted as $\mathcal{R}^{(1)}_{\text{ach}}$ and $\mathcal{R}^{(2)}_{\text{ach}}$,< respectively. The sum of these two regions then results in an achievable region $\mathcal{R}_{\text{ach}}$ for the overall system. We begin by considering the two subsystems and the corresponding achievable rate regions. The first system is given by the equations
\begin{IEEEeqnarray}{rCl}
Y_1^{(1)} & = & S^{q^{(1)}-n_{1}^{(1)}} X_1^{(1)} \oplus S^{q^{(1)}-n_{1}^{(1)}}X_2^{(1)} \oplus S^{q^{(1)}-n_{D}}X_3^{(1)},\quad\,\,\\ \notag
Y_2^{(1)} & = & S^{q^{(1)}-n_{3}^{(1)}} X_3^{(1)},
\end{IEEEeqnarray}
where $n_{1}^{(1)} = n_2-n_M, n_{3}^{(1)} = n_3 - n_M$ and $q^{(1)} = \text{max}\{n_{1}^{(1)},n_{3}^{(1)}\}$.
The second system is defined by
\begin{IEEEeqnarray}{rCl}
Y_1^{(2)} & = & S^{q^{(2)}-n_{1}^{(2)}} X_1^{(2)} \oplus S^{q^{(2)}-n_{2}^{(2)}}X_2^{(2)}, \\ \notag
Y_2^{(2)} & = & S^{q^{(2)}-n_{M}} X_1^{(2)} \oplus S^{q^{(2)}-n_{M}}X_2^{(2)} \oplus S^{q^{}-n_{M}}X_3^{(2)}, 
\end{IEEEeqnarray}
with $q^{(2)} = n_{1}^{(2)} = n_M + \Delta$ and $n_{2}^{(2)} = n_M$. Note that this split is possible due to the weak interference condition $n_D + n_M \leq \text{min}\left\{n_2,n_3\right\}$. The corresponding transmitters and receiver are denoted as $Tx_1^{(1)}, Rx_1^{(1)}$ etc. 
Figure \ref{fig:systemsplit} illustrates the split of the system. Here, the bars represent the bit vectors as seen at the two receivers;  the zero parts due to the channel shifts are not displayed.

We define $\mathcal{R}^{(1)}_{\text{ach}}$ as the set of points $\left(R^{(1)}_1,R^{(1)}_2, R^{(1)}_3\right)$ satisfying
\begin{IEEEeqnarray}{rCl}
R^{(1)}_3 & \leq & n^{(1)}_3 \\
R^{(1)}_1 + R^{(1)}_2 & \leq & n^{(1)}_1 \\
R^{(1)}_1 + R^{(1)}_2 + R^{(1)}_3 & \leq & n_2 + n_3 -n_D - 2n_M.
\end{IEEEeqnarray}
Similarily, $\mathcal{R}^{(2)}_{\text{ach}}$ is defined by the set of equations
\begin{IEEEeqnarray}{rCl}
R^{(2)}_1 + R^{(2)}_2 & \leq & n_{1}^{(2)}\\
R^{(2)}_1 + R^{(2)}_3  & \leq & n_{1}^{(2)} \\
R^{(2)}_2 + R^{(2)}_3 & \leq &  n_M\\
R^{(2)}_1 + R^{(2)}_2 + R^{(2)}_3& \leq & n_M + \varphi(n_M,\Delta).
\end{IEEEeqnarray}
Here, the function $\varphi$ for $p,q \in \mathbb{N}_{0}$ is defined as
\begin{equation}
\varphi(p,q) = \begin{cases} q + \frac{l(p,q) q}{2}, & \mbox{if} \;\; l(p,q) \;\; \mbox{is even} \\ p - \frac{(l(p,q)-1)q}{2}, & \mbox{if} \;\; l(p,q) \;\; \mbox{is odd}. \end{cases}
\end{equation}
where $l(p,q) = p \text{~div~} q$ (note that $0$ is considered an even number). 

\begin{figure}
\begin{center}
\includegraphics[scale=0.54]{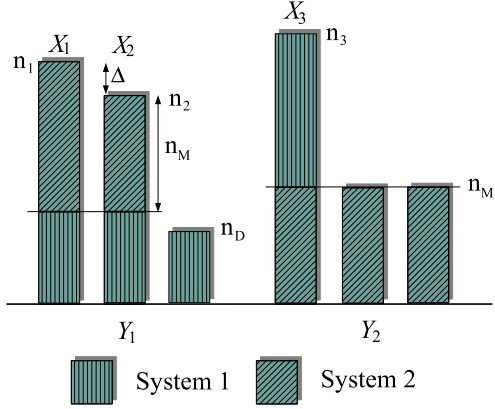}
\caption{Illustration of the split of the system.}
\label{fig:systemsplit}
\end{center} 
\end{figure}

\begin{figure}
\begin{center}
\includegraphics[scale=0.54]{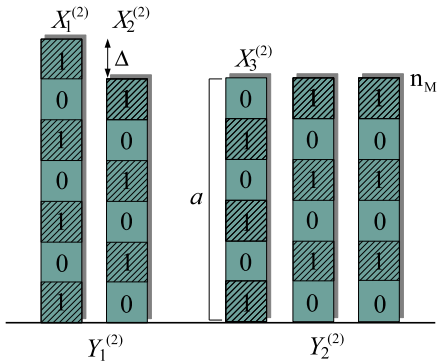}
\caption{Example for optimal bit level assignment: The MAC interference aligns at receiver $Rx_{2}^{(2)}$.}
\label{fig:alignmentexample}
\end{center} 
\end{figure} 

We now show that the rate regions $\mathcal{R}^{(1)}_{\text{ach}}$ and $\mathcal{R}^{(2)}_{\text{ach}}$ are achievable in the respective subsystems. Both regions can be achieved by an orthogonal bit level assignment (and possibly time-sharing). More precisely, for each transmitter, a set of bit levels to be used for data transmission is specified such that at the intended receiver, there is no overlap of these levels with levels used by any other transmitter.

In this way, the achievability of the points in $\mathcal{R}^{(1)}_{\text{ach}}$ follows directly from the results on the interference channel:
If we first consider the one-sided interference channel obtained by removing transmitter $Tx_{2}^{(1)}$, the points $\left(R_{\Sigma}^{(1)},R_3^{(1)}\right)$ with $R_3^{(1)}  \leq  n^{(1)}_3, R_{\Sigma}^{(1)} \leq n^{(1)}_1, R_{\Sigma}^{(1)} + R^{(1)}_3  \leq  n_2 + n_3 -n_D - 2n_M$ can be achieved in this one-sided interference channel using orthogonal coding, which follows from the results in \cite{BreslerTseETCOMM08}. The same rates can be achieved for the interference channel obtained by removing $Tx_{1}^{(1)}$. 
Then, since the signals $X_1^{(1)}$ and $X_2^{(1)}$ are shifted by the same amount at the first receiver, it is clear that in system 1, we can achieve $R_1^{(1)}$ and $R_1^{(2)}$ such that 
$R_1^{(1)} + R_2^{(1)} \leq R_{\Sigma}^{(1)}$, which implies the achievability of region $\mathcal{R}^{(1)}_{\text{ach}}$. 

To show the achievability of $\mathcal{R}^{(2)}_{\text{ach}}$, let $a \in \mathbb{F}_2^{n_M}$ specify the levels used for encoding message $W_{23}$ (note that it suffices to specify the first $n_M$ levels since the others are erased by the channel anyway), where $a_i = 1$ if level $i$ is used and $a_i = 0$ otherwise. Define $\gamma(a) := \mathbf{1}_{n_M} - a$, $\gamma_1(a) := [\gamma(a); \mathbf{1}_{\Delta}]$ and $\gamma_2(a) := [ \mathbf{0}_{\Delta}; \gamma(a)]$. 
Then we can achieve $R_3^{(2)} = |a|$ and all rates $R_1^{(2)}, R_2^{(2)}$ with $R_1^{(2)} \leq |\gamma_1(a)| = n_M + \Delta - |a|, R_2^{(2)} \leq |\gamma_2(a)| = n_M - |a|$ and $R_1^{(2)} + R_2^{(2)} \leq |\gamma_1(a)| + |\gamma_2(a)| - \rho(|a|)$, where
$\rho(x) := \underset{a \in \mathbb{F}_2^{n_M}: |a| = x}{\text{min}}~ \gamma_1(a)^T \gamma_2(a).$
Here, $\gamma_1(a)$ and $[\gamma(a); \mathbf{0}_{\Delta}]$ specify the levels used by transmitters $Tx_1^{(2)}$ and $Tx_2^{(2)}$, respectively, except for the $\rho(|a|)$ positions that cause an overlap of a level in the received bit vectors $\gamma_1(a)$ and $\gamma_2(a)$, for which one of the positions is exclusively assigned to $Tx_1^{(2)}$ or $Tx_2^{(2)}$. 
An assignment vector $a$ solving the minimization for $\rho$ for a given $x$ can be shown to be of the form described in the following. Let $l= n_M \text{~div~} \Delta $ and $Q = n_M \text{~mod~} \Delta$, i.e. $n_M = l\Delta + Q$. We subdivide $a$ into $l \Delta$ subsequences (blocks) of length $\Delta$ and one remainder block of length $Q$. We distribute ones over $a$ until $x$ entries in $a$ have been set to $1$: We start with the even-numbered blocks, followed by the remainder block. If $l$ is even, we finally distribute over the odd-numbered blocks. If $l$ is odd, we also fill the odd-numbered blocks, but in reversed (decreasing) order. To be precise, we define for the case that $l$ is even 
\begin{IEEEeqnarray}{rCl}
A_{\text{even}} &=& \left(\mathbf{0}_{1\times \frac{l}{2}};e_{k}\right)_{k=1}^{l/2}  \otimes I_{\Delta},\\
A_{\text{odd}} &=& \left(e_{k};\mathbf{0}_{1\times \frac{l}{2}}\right)_{k=1}^{l/2} \otimes I_{\Delta}
\end{IEEEeqnarray}
and 
\begin{IEEEeqnarray}{rCl}
A_{\text{even}} &=& \left[\left(\mathbf{0}_{1\times \frac{(l-1)}{2}};e_{k}\right)_{k=1}^{(l-1)/2};\mathbf{0}_{1 \times \frac{l-1}{2}} \right] \otimes I_{\Delta},\\
A_{\text{odd}} &=& M_{l\Delta} \left(\left[\left(e_{k};\mathbf{0}_{1\times \frac{(l+1)}{2}}\right)_{k=1}^{(l-1)/2};e_{\frac{l+1}{2}} \right] \otimes I_{\Delta}\right)
\end{IEEEeqnarray}
for odd $l$. Here, $\otimes$ denotes the Kronecker product, $e_k$ the unit row vector of appropriate size with $1$ at position $k$ and $M_N = (e_{N-k+1})_{k=1}^{N}$ is the flip matrix. Then, defining $P = \left[A_{\text{even}} |\mathbf{0}_{l\Delta\times Q} |A_{\text{odd}}  \right],$
we obtain an optimal assignment $a$ by setting $a = P [\mathbf{1}_{x};\mathbf{0}_{n_M - x}]$. 

We remark that this assignment is not the unique optimal one. Furthermore, it can be interpreted as interference alignment at the receiver $Rx_{2}^{(2)}$: 
The bit levels are chosen such that the interference caused by $X^{(2)}_1$ and $X^{(2)}_2$ (MAC interference) aligns at $Rx_2^{(2)}$ as much as possible in the levels that are unused by $X^{(2)}_3$. Figure \ref{fig:alignmentexample} displays the optimal assignment for the case $l = 6, Q = 0$ and $x = R_3^{(2)} = 3\Delta$. The rates achieved here for the first two transmitters are $R_1^{(2)} = 4 \Delta $ and $R_2^{(2)} = 3 \Delta$. The optimal assignment described above for the case $l$ even results in the following form of the function $\rho$: 
\begin{equation}
\rho(x) = \begin{cases} n_M-2x, & \mbox{if} \;\; 0 \leq x \leq \frac{l\Delta}{2} \\ 
Q + \frac{l \Delta}{2}-x, & \mbox{if} \;\;\frac{l\Delta}{2} \leq x \leq \frac{l\Delta}{2} + Q\\
0, & \mbox{if} \;\; \frac{l\Delta}{2} + Q \leq x \leq n_M.
\end{cases}
\end{equation}
In the case that $l$ is odd, we similarily obtain
\begin{equation}
\rho(x) = \begin{cases} n_M-2x, & \mbox{if} \;\; 0 \leq x \leq Q + \frac{\Delta(l-1)}{2} \\ 
\frac{\Delta(l+1)}{2}-x, & \mbox{if} \;\;Q + \frac{\Delta(l-1)}{2} \leq x \leq \frac{\Delta(l+1)}{2} \\
0, & \mbox{if} \;\; \frac{\Delta(l+1)}{2} \leq x \leq n_M 
\end{cases}
\end{equation}
From this, it is easy to verify that all points in $\mathcal{R}^{(2)}_{\text{ach}}$ are achievable. 

Finally, we obtain an achievable region $\mathcal{R}_{\text{ach}} = \mathcal{R}^{(1)}_{\text{ach}} + \mathcal{R}^{(2)}_{\text{ach}}$ for the overall system. This sum region can be, for example, computed by employing Fourier-Motzkin elimination. The resulting region is stated in the following Proposition:

\begin{prop}
An achievable region $\mathcal{R}_{\text{ach}}$ is given by the set of points $(R_1,R_2,R_3) \in \mathbb{R}^3$ satisfying the constraints
\begin{IEEEeqnarray}{rCl} \label{eq:bound1}
R_1 & \leq & n_1 \\ \label{eq:bound2}
R_2 & \leq & n_2 \\ \label{eq:bound3}
R_3 & \leq & n_3 \\ \label{eq:bound4}
R_1 + R_2 & \leq & n_1 \\ \label{eq:bound5}
R_1 + R_3 & \leq & n_1 + n_3 - n_D - n_M \\ \label{eq:bound6}
R_2 + R_3 & \leq & n_2 + n_3 - n_D - n_M\\ \label{eq:bound7}
R_1 + R_2 + R_3 & \leq & n_2 + n_3 - n_D - n_M + \varphi(n_M,\Delta)~\\ \label{eq:bound8}
R_1 + R_2 + 2R_3 & \leq & n_1 + 2n_3 - n_D - n_M.
\end{IEEEeqnarray}
\end{prop}

\section{Capacity Region}\label{sec:capacity} 
In this section, we show that the achievable region $\mathcal{R}_{\text{ach}}$ from the previous section actually constitutes the capacity region $\mathcal{C}$.
For the proof, we need the following Lemma: 
\begin{lem}\label{lem:entropybound}
For two independent random matrices $A,B \in \mathbb{F}_2^{n + \Delta \times m}$ with $m,n, \Delta \in \mathbb{N}$, it holds that
\begin{IEEEeqnarray}{rCl} \label{eq:shiftbound1}
H(A \oplus S^{\Delta} B) - H(S^{\Delta} A \oplus S^{\Delta} B) &\leq& m \varphi(n,\Delta),\\\label{eq:shiftbound2}
H(A \oplus S^{\Delta} B) - 2 H(S^{\Delta} A \oplus S^{\Delta} B) &\leq&  m \Delta.
\end{IEEEeqnarray}
\end{lem}

\begin{proof}
In order to show (\ref{eq:shiftbound1}), we let $l = n ~\text{div}~ \Delta, Q = n~ \text{mod}~ \Delta$ and introduce the following labels for blocks of rows of the matrices $A$ and $B$:
$A = [S;(A_k)_{k=1}^{l+1}], B = [T;(B_k)_{k=1}^{l+1}]$, where $S,T \in \mathbb{F}_2^{Q \times m}$ and $A_k,B_k \in \mathbb{F}_2^{\Delta \times m}$. 
Then the shifted versions of $A$ and $B$ are $S^{\Delta} A = [\mathbf{0}_{\Delta \times m}; S;(A_k)_{k=1}^{l}], S^{\Delta} B = [\mathbf{0}_{\Delta \times m};T;(B_k)_{k=1}^{l}].$
First consider the case that $l$ is even. Then we have
\begin{IEEEeqnarray}{l}
H(A \oplus S^{\Delta} B) - H(S^{\Delta} A \oplus S^{\Delta} B) \\ \notag
\,\, \,\,\,= H\left[S; A_1 \oplus [\mathbf{0}_{\Delta-Q \times m};T];(A_{k+1} \oplus B_k)_{k=1}^{l} \right] \\ \notag
\,\,\,\,\,\,\,\,\,\,\,- H\left[S \oplus T;(A_k \oplus B_k)_{k=1}^{l}\right] \\ \notag
\,\, \,\,\,\leq H\left[S; A_1 \oplus [\mathbf{0}_{\Delta-Q  \times m};T];(A_{k+1} \oplus B_k)_{k=1}^{l}\right] \\ \notag
\,\,\,\,\,\,\,\,\,\,\,- H\left[S \oplus T;(A_k \oplus B_k)_{k=1}^{l}\left|T,(A_{2k-1})_{k=1}^{l/2},(B_{2k})_{k=1}^{l/2}\right.\right]
\\ \notag
\,\, \,\,\,= H\left[S; A_1 \oplus [\mathbf{0}_{\Delta-Q  \times m};T];(A_{k+1} \oplus B_k)_{k=1}^{l}\right] \\ \notag
\,\,\,\,\,\,\,\,\,\,\,- H\left[S; (A_{2k})_{k=1}^{l/2}; (B_{2k-1})_{k=1}^{l/2}\right]
\\ \notag
\,\, \,\,\,\leq Ã‚Â¸m \Delta \left(1 + \frac{l}{2}\right) + H\left[S; \left(A_{2k} \oplus B_{2k-1}\right)_{k=1}^{l/2}\right]\\ \notag
\,\,\,\,\,\,\,\,\,\,\, - H\left[S; (A_{2k})_{k=1}^{l/2}; (B_{2k-1})_{k=1}^{l/2}\right] \\ \notag
\,\, \,\,\,\leq m \Delta \left(1 + \frac{l}{2}\right) = m \varphi(n,\Delta).
\end{IEEEeqnarray}
where the last inequality is due to the independence of $A$ and $B$.

For $l$ odd, we can bound the expression as follows:
\begin{IEEEeqnarray}{l}
H(A \oplus S^{\Delta} B) - H(S^{\Delta} A \oplus S^{\Delta} B) \\ \notag
\,\, \,\,\,\leq H\left[S; A_1 \oplus [\mathbf{0}_{\Delta-Q  \times m};T];(A_{k+1} \oplus B_k)_{k=1}^{l}\right] \\ \notag
\,\,\,\,\,\,\,\,\,\,\,- H\left[S \oplus T;(A_k \oplus B_k)_{k=1}^{l}\left|S,(A_{2k})_{k=1}^{(l-1)/2},\right.\right.\\\notag
\,\,\,\,\,\,\,\,\,\,\,\,\,\,\,\,\,\,\,\,\,\,\,\,\,\,\,\,\,\,\, \left.(B_{2k-1})_{k=1}^{(l+1)/2}\right]\\ \notag
\,\, \,\,\,= H\left[S; A_1 \oplus [\mathbf{0}_{\Delta-Q  \times m};T];(A_{k+1} \oplus B_k)_{k=1}^{l}\right] \\ \notag
\,\,\,\,\,\,\,\,\,\,\,- H\left[T; (A_{2k-1})_{k=1}^{(l+1)/2}; (B_{2k})_{k=1}^{(l-1)/2}\right]
\\ \notag
\,\, \,\,\,\leq Ã‚Â¸m \frac{\Delta(l+1)}{2} + mQ \\\notag
\,\,\,\,\,\,\,\,\,\,\, + H\left[A_1 \oplus [\mathbf{0}_{\Delta-Q  \times m};T];\left(A_{2k+1} \oplus B_{2k} \right)_{k=1}^{(l-1)/2}\right]\\ \notag
\,\,\,\,\,\,\,\,\,\,\, - H\left[T; (A_{2k-1})_{k=1}^{(l+1)/2}; (B_{2k})_{k=1}^{(l-1)/2}\right] \\ \notag
\,\, \,\,\,\leq Ã‚Â¸m  \frac{\Delta(l+1)}{2}+ mQ = m \varphi(n,\Delta).
\end{IEEEeqnarray}

The bound (\ref{eq:shiftbound2}) follows from 
\begin{IEEEeqnarray}{l}
H(A \oplus S^{\Delta} B) - 2 H(S^{\Delta} A \oplus S^{\Delta} B) \\ \notag
\,\, \,\,\,\leq H(A) + H(S^{\Delta} B) - 2 H(S^{\Delta} A \oplus S^{\Delta} B)\\ \notag
\,\, \,\,\,\leq m \Delta +  H(S^{\Delta} A)+ H(S^{\Delta} B) - 2 H(S^{\Delta} A \oplus S^{\Delta} B)\\ \notag
\,\, \,\,\,= m\Delta +  H(S^{\Delta} A \oplus S^{\Delta} B) \\ \notag
\,\,\,\,\,\,\,\,\,\, + H(S^{\Delta} A|S^{\Delta} A \oplus S^{\Delta} B)- 2 H(S^{\Delta} A \oplus S^{\Delta} B)\\ \notag
\,\, \,\,\,= m\Delta  + H(S^{\Delta} A|S^{\Delta} A \oplus S^{\Delta} B)-H(S^{\Delta} A \oplus S^{\Delta} B)  \\ \notag
\,\, \,\,\,= m\Delta  + H(S^{\Delta} A|S^{\Delta} A \oplus S^{\Delta} B) - H(S^{\Delta} A) \\ \notag
\,\,\,\,\,\,\,\,\,\, + H(S^{\Delta} B|S^{\Delta} A \oplus S^{\Delta} B) - H(S^{\Delta} B) \\ \notag
\,\, \,\,\,\leq m\Delta.
\end{IEEEeqnarray}

\end{proof}
We are now ready to prove the main result of the paper:
\vspace{0.2cm}
\begin{thm}
The capacity region $\mathcal{C}$ is given by $\mathcal{R}_{\text{ach}}$.
\end{thm}
\vspace{0.2cm}
\begin{proof}
Consider the interference channel formed by $Tx_{1}, Tx_{3},Rx_{1}$ and $Rx_{2}$ with corresponding capacity region $\mathcal{C_{\text{IC1/3}}}$, the interference channel build from $Tx_{2}, Tx_{3},Rx_{1}$ and $Rx_{2}$ with capacity region $\mathcal{C_{\text{IC2/3}}}$ and the multiple-access channel consisting of $Tx_{1}, Tx_{2}$ and $Rx_{1}$ with capacity region $\mathcal{C}_{\text{MAC}}$. 
Then, it is clear that if $(R_1,R_2,R_3) \in \mathcal{C}$, we must have $(R_1,R_2) \in \mathcal{C}_{\text{MAC}}, (R_1,R_3) \in \mathcal{C_{\text{IC1/3}}}$ and $(R_2,R_3) \in \mathcal{C_{\text{IC2/3}}}$. Evaluating the corresponding capacity regions \cite{BreslerTseETCOMM08}, this implies the bounds (\ref{eq:bound1}) - (\ref{eq:bound6}).

In order to prove (\ref{eq:bound7}), we apply Fano's inequality: For each triple of achievable rates $(R_1,R_2,R_3) \in \mathcal{C}$, Fano's inequality implies that there exists a sequence
$\varepsilon_{N}$ with $\varepsilon_{N} \rightarrow 0$ for $N \rightarrow \infty$ and
a sequence of joint factorized distributions $p(x_1^N)p(x_2^N)p(x_3^N)$ such that for all $N \in \mathbb{N}$
\begin{IEEEeqnarray}{rCl}
R_1 + R_2 &\leq& \frac{1}{N}I(X_1^N,X_2^N;Y_1^N) + \varepsilon_{N},\\
R_3 &\leq& \frac{1}{N}I(X_3^N;Y_2^N) + \varepsilon_{N}.
\end{IEEEeqnarray}
We let $\beta = n_2 - n_D-n_M, \epsilon = n_3 - n_M - n_D$, $T_{k} = [\mathbf{0}_{q-k \times m};\mathbf{1}_{k \times m}]$ and define the following partial matrices of the components of the received signals (note that  $\beta, \epsilon \geq 0$ from the weak interference condition):
\begin{IEEEeqnarray}{lllllll} \notag
\hat{X}_{1} &=& S^{q-n_M}X_{1}^N, X_{1}^{\uparrow} &=& S^{q-n_M-\Delta}X_{1}^N, X_{1}^{\downarrow} &=& T_{n_D+\beta}X_{1}^N,  \\ \notag
\hat{X}_{2} &=& S^{q-n_M}X_{2}^N, X_{2}^{\downarrow} &=& T_{n_D+\beta}X_{2}^N, & &\\ \notag
\hat{X}_{3} &=& S^{q-n_D}X_{3}^N, X_{3}^{\uparrow} &=& S^{q-n_D-\epsilon}X_{3}^N, X_{3}^{\downarrow} &=& T_{n_M}X_{3}^N.
\end{IEEEeqnarray}
Then we have the following chain of inequalities:
\begin{IEEEeqnarray}{l}
(R_1 + R_2 + R_3)N -\varepsilon_{N}N \\ \notag
\,\, \,\,\,\leq I(X_1^N,X_2^N;Y_1^N) + I(X_3^N;Y_2^N) \\ \notag
\,\,\,\,\,= H(Y_1^N) - H(Y_1^N|X_1^N,X_2^N) + H(Y_2^N) - H(Y_2^N|X_3^N)\\ \notag
\,\,\,\,\, = H(Y_1^N) - H(\hat{X}_{3}) + H(Y_2^N) - H(\hat{X}_{1} \oplus \hat{X}_{2})\\ \notag
\,\,\,\,\,\leq H(X_{1}^{\uparrow} \oplus \hat{X}_{2}) + H(X_{1}^{\downarrow} \oplus X_{2}^{\downarrow} \oplus \hat{X}_3) - H(\hat{X}_{3}) \\ \notag 
\,\,\,\,\,\,\,\,\,\,\,+  H(X_{3}^{\uparrow}) + H(X_{3}^{\downarrow} \oplus \hat{X}_{1} \oplus \hat{X}_{2}) - H(\hat{X_1} \oplus \hat{X_2}) \\
\,\,\,\,\, \stackrel{\text{\small(a)}}{\leq}  H(X_{1}^{\uparrow} \oplus \hat{X}_{2}) - H(X_{1}^{\downarrow} \oplus X_{2}^{\downarrow}) \\ \notag
\,\,\,\,\,\,\,\,\,\,\,+ (n_2 + n_3 - n_M - n_D)N\\ \notag
\,\,\,\,\, \stackrel{\text{\small(b)}}{\leq} N\varphi(n_M,\Delta) + (n_2 + n_3 - n_M - n_D)N 
\end{IEEEeqnarray}
where $(a)$ follows from 
\begin{IEEEeqnarray}{rCl}
 H(X_{3}^{\downarrow} \oplus \hat{X}_{1} \oplus \hat{X}_{2}) & \leq & n_M, \\
 H(X_{1}^{\downarrow} \oplus X_{2}^{\downarrow} \oplus \hat{X}_3) & \leq & n_2 - n_M,\\
 H(X_{3}^{\uparrow}) - H(\hat{X}_{3}) & \leq & n_3 - n_M - n_D
\end{IEEEeqnarray}
and $(b)$ is obtained by applying Lemma \ref{lem:entropybound}. A similar argument, using the second part of Lemma \ref{lem:entropybound}, shows the bound (\ref{eq:bound8}). 
\end{proof}

Figure \ref{fig:example} shows the capacity region for the case  $n_1 = 18$,$n_2 = 16$,$n_3 = 14$,$n_M = 6$ and $n_D = 7$.

\begin{figure}
\begin{center}
\includegraphics[scale=0.5]{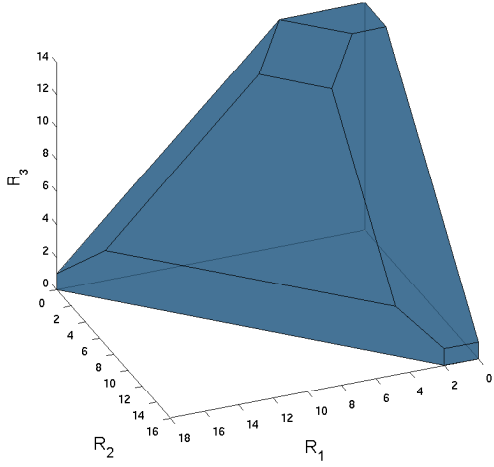}
\caption{Capacity region of the system for $n_1 = 18$,$n_2 = 16$,$n_3 = 14$,$n_M = 6$ and $n_D = 7$.}
\label{fig:example}
\end{center} 
\end{figure} 

\enlargethispage{-0.7cm}

\section{Conclusions}\label{sec:conclusions} 
In this paper, we studied the linear deterministic model for a cellular-type channel where there are two users (in cell 1) transmitting to a receiver (base station 1), mutually interfering with a third transmitter (in cell 2) communicating with a second base station (base station 2). We studied the case of symmetric weak interference where the interference links from cell 1 user to the cell 2 base station are identical and the sum of the interference gains are less or equal then the smallest direct link. We derived the capacity region and the corresponding transmission scheme. 
Even though the system resembles the interference channel, the presence of the second link in cell 1 offers additional potential for aligning the interference caused by the two users at the receiver in the second cell. For this, the transmitted signals in cell 1 are chosen such that the interference at the second receiver aligns on the part of the signal that is unused by the transmitter in cell 2 as much as possible. Although we have considered a restricted setup in this paper, we believe that the achievability and converse arguments used in this paper give valuable insights for the consideration of more general systems.

Future work will study extensions to the more general case of more users and less restricted channel gains. Another interesting direction for further investigations are the connections to the Gaussian equivalent of the channel, specifically concerning approximate capacity characterizations and / or the determination of (generalized) degrees of freedom of the system.

\bibliographystyle{/home/buehler/TeX/BibTex/IEEEtran}


\end{document}